\documentclass[11pt]{article}

\usepackage{amsmath,amssymb,graphicx,color,amsthm}
\usepackage[linesnumbered,ruled]{algorithm2e}
\usepackage{tikz}
\usepackage{svg}
\usepackage{array}

\usepackage{mathtools}
\usetikzlibrary{shapes.geometric, arrows}
\usepackage{algorithmic}
\usepackage[english]{babel}
\usepackage[utf8]{inputenc}
\usepackage{subcaption,afterpage} 
\usepackage{mathtools}
\usepackage{amsfonts}
\usepackage[colorinlistoftodos]{todonotes}
\usepackage{enumerate}
\usepackage[round]{natbib}
\usepackage{hyperref}
\newcommand{\excise}[1]{}

\newtheorem{manualtheoreminner}{Theorem}
\newenvironment{manualtheorem}[1]{%
  \renewcommand\themanualtheoreminner{#1}%
  \manualtheoreminner
}{\endmanualtheoreminner}

\renewcommand\_{\textunderscore\allowbreak}

\newtheorem{theorem}{Theorem}

\newtheorem{lemma}{Lemma}

\DeclareMathOperator*{\argmax}{argmax}

\DeclareMathOperator\rank{rank}
\DeclareMathOperator\corank{corank}

\newcommand{\RNum}[1]{\uppercase\expandafter{\romannumeral #1\relax}}

\usepackage{setspace}
\usepackage{lineno}

\begin{document}
	
%
	\title{\mbox{}\\[-11ex]Tree Representations of Brain Structural  Connectivity via Persistent Homology}
	\author{\vspace{1em}Didong Li$^{1}$, Phuc Nguyen$^2$, Zhengwu Zhang$^3$ and David B Dunson$^{2}$ \\ 
	{\em Department of Biostatistics$^1$ and Department of Statistics and Operations Research$^3$}\\
	{\em University of North Carolina at Chapel Hill}\\
			{\em Department of Statistical Science$^2$, Duke University}}
	\date{\vspace{-5ex}}
	\maketitle
%


\begin{abstract}
The brain structural connectome is generated by a collection of white matter fiber bundles constructed from diffusion weighted MRI (dMRI), acting as highways for neural activity.  There has been abundant interest in studying how the structural connectome varies across individuals in relation to their traits, ranging from age and gender to neuropsychiatric outcomes.  After applying tractography to dMRI to get white matter fiber bundles, a key question is how to represent the brain connectome to facilitate statistical analyses relating connectomes to traits.  The current standard divides the brain into regions of interest (ROIs), and then relies on an {\em adjacency matrix} (AM) representation.  Each cell in the AM is a measure of connectivity, e.g., number of fiber curves, between a pair of ROIs.  Although the AM representation is intuitive, a disadvantage is the high-dimensionality due to the large number of cells in the matrix.  This article proposes a simpler tree representation of the brain connectome, which is motivated by ideas in computational topology and takes topological and biological information on the cortical surface into consideration.  We demonstrate that our tree representation preserves useful information and interpretability, while reducing dimensionality to improve statistical and computational efficiency.  Applications to data from the Human Connectome Project (HCP) are considered and code is provided for reproducing our analyses.
\end{abstract}
\noindent \textbf{Key Words}: 
Adjacency matrix, brain connectome, persistent homology, structural connectivity, tree.

\section{Introduction}
\label{sec:intro}
The human brain structural connectome, defined as the white matter fiber tracts connecting different brain regions, plays a central role in understanding how brain structure impacts human function and behavior \citep{park2013review}. Recent advances in neuroimaging methods have led to increasing collection of high quality functional and structural connectome data in humans.  There are multiple large datasets available containing 1,000s of connectomes, including the Human Connectome Project (HCP) and the UK Biobank \citep{essen2012hcp, bycroft2018biobank}. We can now better relate variations in the connectomes between individuals to phenotypic traits \citep{wang2012alzheimer, hong2019atypical, roy2019parkinson}. However, the large amount of data also creates the need for informative and efficient representations of the brain and its structural connectome \citep{gelletta2017storage, zhang2018pipeline, pizarro2019database, zalesky2019buildconnectomes, jeurissen2017tractography}. The main focus of this article is a novel and efficient representation of the processed connectome data as an alternative to the adjacency matrix (AM) as input to statistical analyses.

Diffusion magnetic resonance imaging (dMRI) uses diffusion-driven displacement of water molecules in the brain to map the organization and orientation of white fiber tracts on a microscopic scale \citep{bihan2003dmri}. Applying tractography to the dMRI data, we can construct a ``tractogram'' of 3D trajectories of white fiber tracts \citep{jeurissen2017tractography}. It is challenging to analyze the tractogram directly because 1) the number of fiber trajectories is extremely large; 2) the tractogram contains geometric structure; and 3) alignment of individual tracts between subjects remains difficult \citep{zhang2019tensor}. Because of these challenges, it is common to parcellate the brain into anatomical regions of interest (ROIs) \citep{desikan2006automated,destrieux2010automatic, he2010graphbase}, and then extract fiber bundles connecting ROIs. We can then represent the brain structural 
connectome as a weighted network in a form of an AM, a $p$ by $p$ symmetric matrix, with $i,j$-th entry equal to the number of fiber curves connecting region $i$ and region $j$, where $p$ is the total number of regions in the parcellation. 

Statistical analyses of structural connectomes are typically based on this AM representation, 
which characterizes the connectome on a fixed scale depending on the resolution of ROIs.  However, 
research has shown that brain networks fundamentally organize as multi-scale and hierarchical entities \citep{bassett2013multiscale}. Some research has attempted to analyze community structures in functional and structural brain networks across resolutions \citep{bassett2017multiscale}; however, these works are limited to community detection. Brain atlases have anatomically meaningful hierarchies but only one level can be captured by the AM representation. If the lowest level in the atlas hierarchy with the greatest number of ROIs is used, this creates a very high-dimensional representation of the brain. The number of pairs of ROIs often exceeds the number of connectomes in the dataset.  This presents statistical and computational challenges, with analyses often having low power and a lack of interpretability 
 \citep{poldrack2017reproduce, cremers2017powerhcp}.

For instance, to infer relationships between the connectome and a trait of interest, it is common to conduct hypothesis tests for association between each edge (connection strength between a pair of ROIs) and the trait \citep{fornito2016book, gou2018test,wang2019symmetric,lee2021test}. As the number of edges is very large, such tests will tend to produce a large number of type I errors without multiple testing adjustment.  If a Bonferroni adjustment is used, then the power for detecting associations between particular edges and traits will be very low.  A common alternative is to control for false discovery rate, for example via the Benjamini-Hochberg approach \citep{nichols2002fdr}. However, such corrections cannot solve the inevitable increase in testing errors that occur with more ROIs.  An alternative is to take into account the network structure of the data in the statistical analysis; see, for example \cite{Leek18718, fornito2016book, ALBERTON2020116760}. Such approaches can potentially improve power to detect differences while controlling type I errors through appropriate borrowing of information across the edges or relaxation of the independence assumption, but statistical and computational problems arise as the number of ROIs increases.
Finally, it is common to vectorize the lower-triangular portion of the brain AM and then apply regression or classification methods designed for high-dimensional features (e.g., by penalizing using the ridge or lasso penalties).

To make the problems more concrete, note that a symmetric $p \times p$ connectome AM has $\frac{(p-1)p}{2}$ pairs of brain ROIs.  For the popular Desikan-Killiany parcellation \citep{desikan2006automated} with $p=68$,
$\frac{(p-1)p}{2} = 2278$.  A number of other common atlases have many more than $p=68$ brain regions, leading to a much larger number of edges.  Even if one is relying on data from a large cohort, such as the UK Biobank, the sample size (number of subjects) is still much smaller than the connectivity features, leading to statistical efficiency problems without reducing the dimensionality greatly from $q=\frac{(p-1)p}{2}$.    Dimensionality reduction methods, such as Principal Components Analysis (PCA), tensor decomposition, or non-negative matrix factorization, can be a remedy for studying relationships between the connectome and traits \citep{yourganov2014pca, smith2015cca, zhang2019tensor,patel2020nmf}.
But, they may lack interpretability and fail to detect important relationships when
the first few principal components are not biologically meaningful or predictive of traits.

In this paper, we propose a new representation of the brain connectome that is inspired by ideas in computational topology, a field focused on developing computational tools for investigating topological and geometric structure in complex data \citep{carlsson2009topology,edelsbrunner2010computational}.  A common technique in computational topology is {\em persistent homology}, which investigates geometry/topology of the data by assessing how features of the data come and go at different scales of representation.  Related ideas have been used successfully in studying brain vascular networks \citep{bendich2016persistent}, hippocampal spatial maps \citep{dabaghian2012topological}, dynamical neuroimaging spatiotemporal representations \citep{geniesse2019generating}, neural data decoding \citep{rybakken2019decoding} and so on \citep{sizemore2019importance}. 
We propose a fundamentally different framework, which incorporates an anatomically meaningful hierarchy of brain regions within a persistent homology approach to produce a new tree representation of the brain structural connectome. This representation reduces dimensionality substantially relative to the AM approach, leading to statistical and computational advantages, while enhancing interpretability. After showing our construction and providing mathematical and biological justification, we contrast the new representation with AM representations in analyses of data from the Human Connectome Project (HCP).

\section{Method}\label{sec:method}	

\subsection{Tree construction}

There is a rich literature defining a wide variety of parcellations of the brain into regions of interest (ROIs) that are motivated by a combination of biological and statistical justifications \citep{desikan2006automated, destrieux2010automatic, klein2012dkt}. The ROIs should ideally be chosen based on biological function and to avoid inappropriately merging biologically and structurally distinct regions of the brain.  Also, it is important to not sub-divide the brain into regions that are so small that (a) it may be difficult to align the data for different subjects and (b) the number of ROIs is so large that the statistical and computational problems mentioned in the introduction are exacerbated.  Based on such considerations, the Desikan-Killiany (DK) atlas is particularly popular, breaking the brain into $p=68$ ROIs \citep{desikan2006automated}.

A parcellation such as DK is typically used to construct an AM representation of the structural connectome at a single level of resolution. However, we instead propose to introduce a multi-resolution tree in which we start with the entire cortical surface of the brain as the root node, and then divide into the right and left hemisphere to produce two children of the root node. We further sub-divide the two hemispheres into large sub-regions, then divide these sub-regions into smaller regions. We continue this sub-division until obtaining the regions of DK (or another target atlas) as the leaf nodes in the tree. In doing this, we note that there is substantial flexibility in defining the tree structure; we need not choose a binary tree and can choose the regions at each level of the tree based on biological function considerations to the extent possible. Finally, we summarize the connectivity information at different resolutions in the weights of the tree nodes.

In this article, we focus primarily on the following tree construction based on the DK atlas for illustration, while hoping that this work motivates additional work using careful statistical and biological thinking to choose the regions at each layer of the tree. The DK parcellation is informed by standard neuroanatomical conventions, previous works on brain parcellations, conversations with expert scientists in neuroscience, and anatomic information on local folds and grooves of the brain \citep{desikan2006automated}. The DK protocol divides each hemisphere into 34 regions that can be organized hierarchically \citep{desikan2006automated}. Specifically, each hemisphere has six regions: frontal lobe, parietal lobe, occipital lobe, cingulate cortex, temporal lobe, and insula. Most of these regions have multiple sub-regions, many of which are further sub-divided in the DK atlas. The full hierarchy can be found in the Appendix. We calculate the weight of each node as the sum of all connections between its immediate children. For instance, the weight at the root node will equal the sum of all the inter-hemisphere connections. The weight of the left hemisphere node will equal the sum of all connections among the left temporal lobe, frontal lobe, parietal lobe, occipital lobe, cingulate cortex, and insula. The weight of the left temporal lobe is the sum of all connections between regions within the temporal medial aspect and lateral aspect. We continue this calculation for all nodes that have children. The weights of a leaf node will be all connections within that region, or equivalently, the diagonal element at that region's index on the AM representation. 

\begin{figure}[ht]
    \centering
    \includegraphics[width=\textwidth, trim={0 50 0 50}, clip]{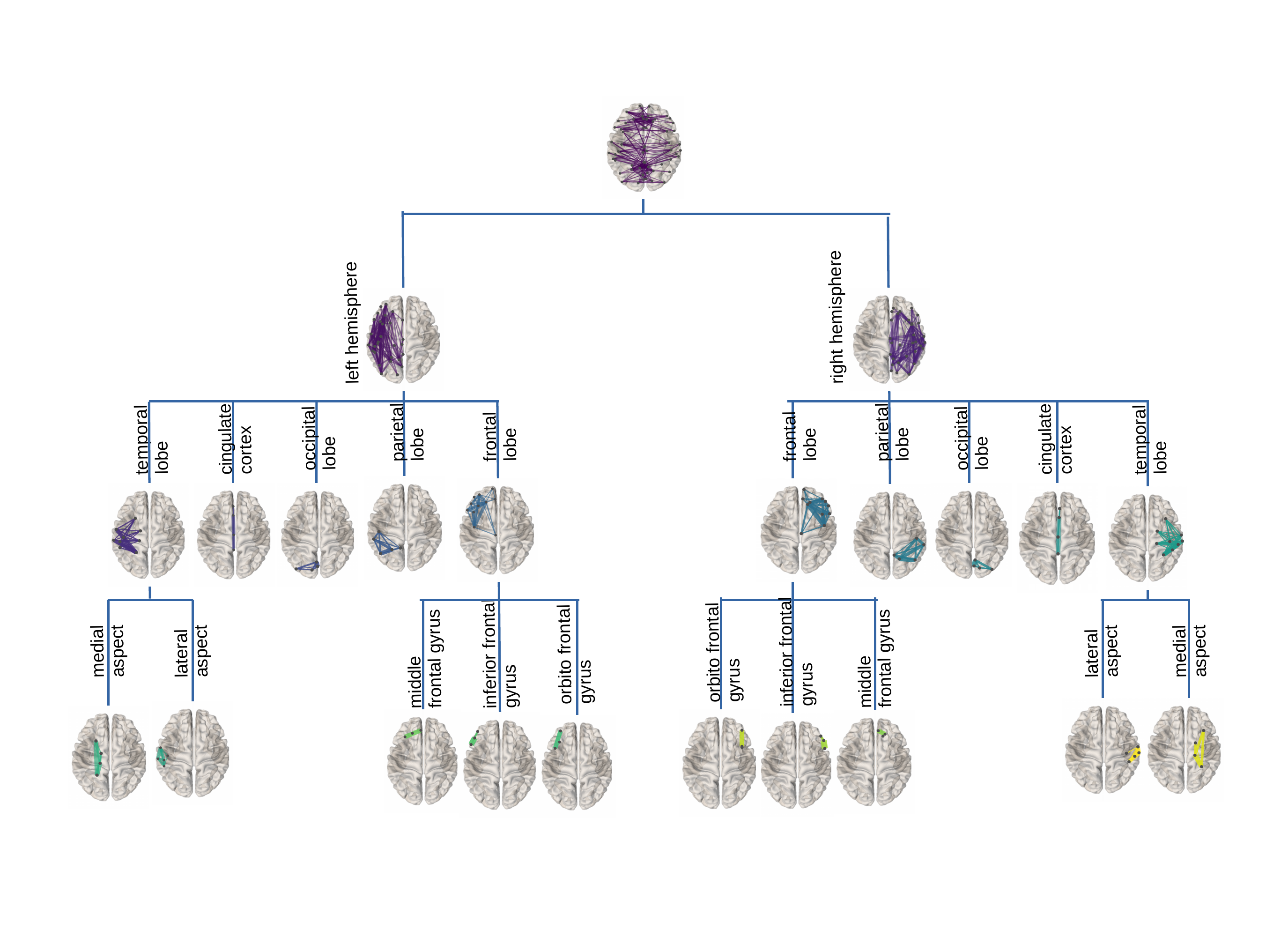}
    \caption{Visualization of the DK tree structure and connections summarized at each node on the brain image. Leaves have no internal connections and are omitted for readability.}
    \label{fig:brain-tree}
\end{figure}

To rephrase the tree construction in math notations, let $A=A^1_1$ be the whole brain, $A^2_1$ be the left hemisphere, $A^2_2$ be the right hemisphere, and $A^l_i$
be the $i$-th region at level $l$, where $l=1,\cdots,L$ and $i=1,\cdots,N_l$ ($L=4$, $N_L=10$ for DK). Define the weight of region $A^l_i$, denoted by $H^l_i$, as $H^{l}_i\coloneqq \text{number of fibers connecting any two children of } A^{l}_i$ for all $l=1,2,\cdots,L$ and $i=1,\cdots,N_l$.
For example, for DK-based tree, $H_1^2$ is the number of fibers connecting children of left hemisphere, that is, fibers between temporal lobe, cingulate cortex, occipital lobe, parietal lobe, and frontal lobe. Similarly, $H_{1}^3$ is the number of fibers connecting children of temporal lobe, that is, fibers between medial aspect and lateral aspect.

    

Figure \ref{fig:brain-tree} provides an illustration of the DK-based tree structure and the connections summarized at each node. Often the connections within the leaf nodes are not estimated in the AM representation (i.e., the diagonal elements of an AM are set to zero), thus, they have weights zero and are omitted from the figure. We also introduce a more compact visualization of the tree based on circle chord plots. Figure \ref{fig:brain-chord} shows the steps in our pipeline to construct a brain tree. The final output circle plot displays bundles of white matter tracts as overlapping chords scaled inversely proportional to the level of the tree they belong to and color-coded by the node they belong to. 

\begin{figure}[ht]
    \centering
    \includegraphics[width=\textwidth, trim={0 10 0 10}, clip]{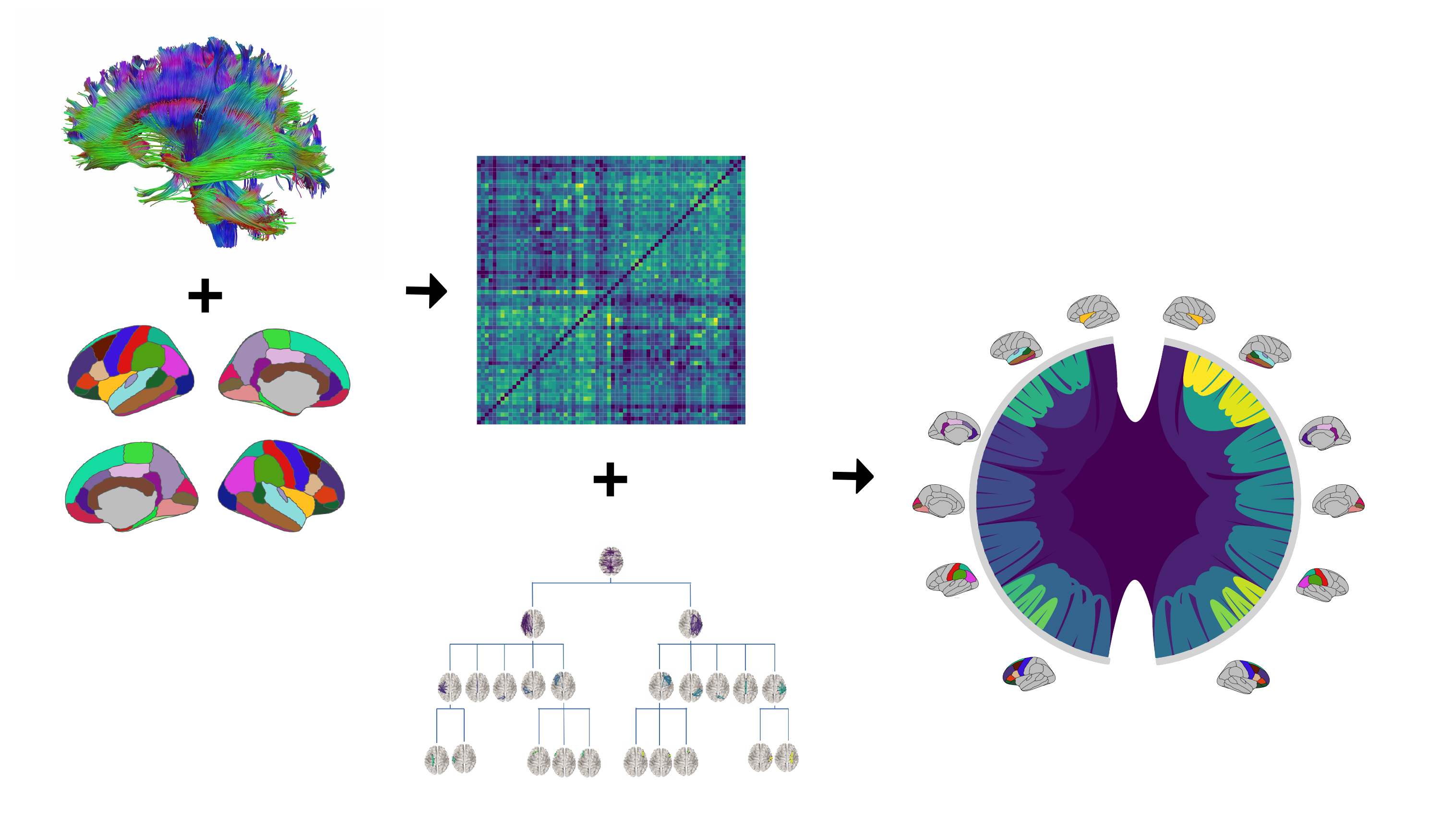}
    \caption{The tree construction pipeline: first, tractography and Desikan-Killiany (DK) protocol are used to estimate an adjacency matrix; then, the adjacency matrix and a hierarchy based on the DK protocol are combined to produce the tree representation. We introduce a compact circle chord plot that shows the tree representation of white matter fiber tracts connecting brain regions from the DK protocol. Connections are color-coded by the node they belong to, the same color codes as in Figure \ref{fig:brain-tree}, and scaled inversely proportional to the level they belong to in the tree.}
    \label{fig:brain-chord}
\end{figure}

\begin{figure}
    \centering
    \includegraphics[width=\textwidth]{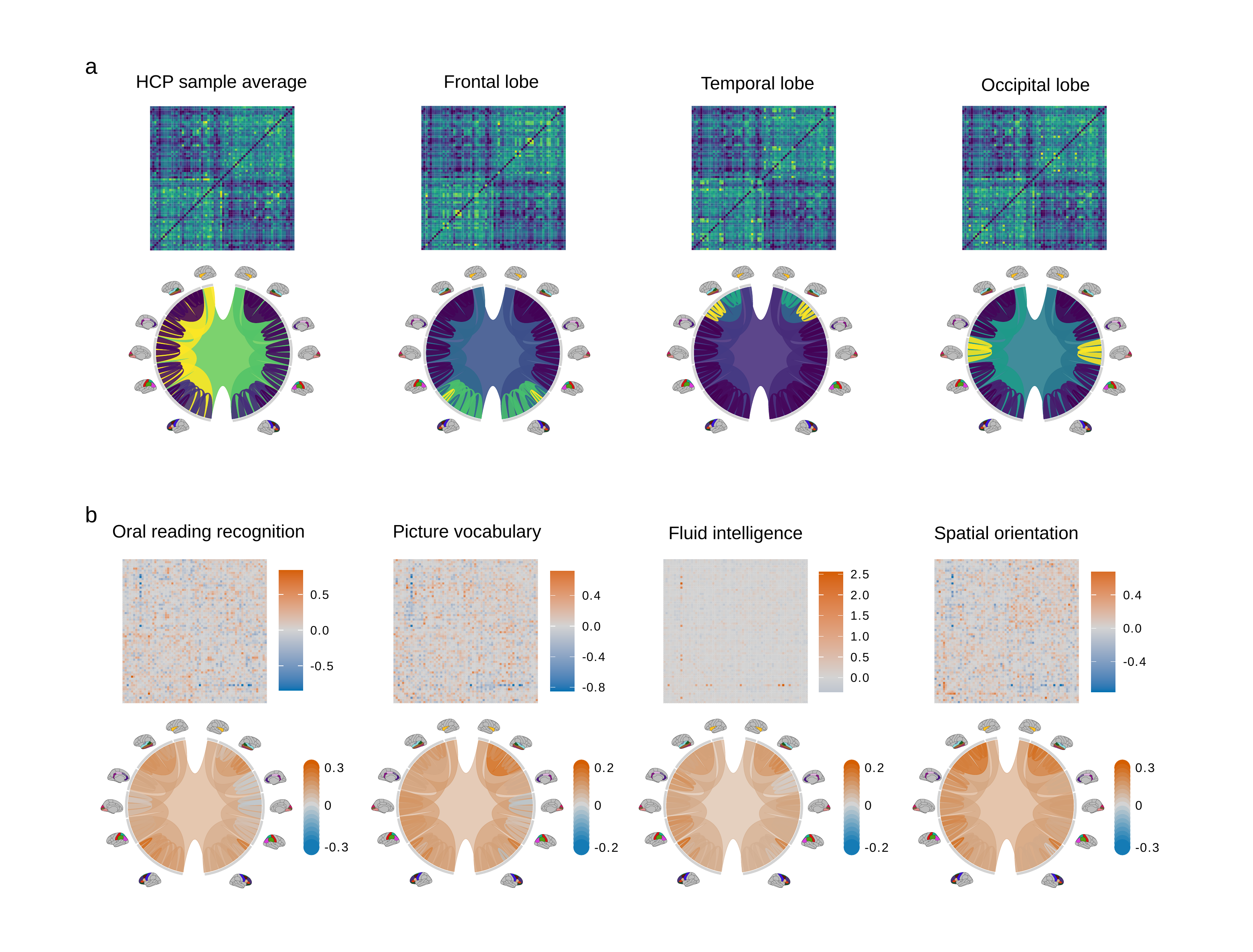}
    \caption{Examples of adjacency matrices and corresponding tree representations of brain connectomes. a) The sample average from the HCP data (left most) and its modifications so that different regions of the brain (i.e. the frontal lobe, temporal lobe, and occipital lobe) have denser connections. Because these modifications are unrealistic and only meant for better visualization of the differences between AM and tree representations, color scales are omitted. Values in the adjacency matrices are normalized for visual clarity. b) Percent difference in brain connections between the top and bottom ten percent of scores in four cognitive tasks.}
    \label{fig:tree-examples}
\end{figure}





\subsection{A Persistent Homology Interpretation}
Persistent homology is a method for computing topological features of a space at different resolutions. As quantitative features of noisy data, persistent homology is less sensitive to the choice of coordinate and metric and robust to noise \citep{carlsson2009topology}. The key construction in persistent homology is the filtration, a multi-scale structure similar to the brain network. As a result, we can interpret the above defined $H^l_k$ as corank of the  persistent homology. As a rigorous definition of persistent homology is highly technical, we present a simple version and leave the rigorous version and the proof to the Appendix. For relevant background in topology, see
 \cite{munkres2016topology,hatcher2002algebraic}. 
\begin{theorem}\label{thm:persistent}
	$H^l_k$ is the corank of the persistent homology.	
\end{theorem}
Theorem \ref{thm:persistent} provides a topological interpretation of $H_k^l$, and partially explains why the tree $T=(A,H)$ is a powerful representation of the brain network.

\section{Results}\label{sec:result}
\subsection{Data description}
We investigate our tree representation's ability to preserve information from the AM representation while improving interpretability in analyses relating brain structures to behavioral traits. We use neuroimaging data and scores on various behavioral assessments from the HCP \citep{glassier2013hcp, glasser2016hcp}. The HCP collects high-quality diffusion MRI (dMRI) and structural MRI (sMRI) data, characterizing brain connectivity of 1200 healthy adults, and enables comparison between brain circuits, behavior and genetics at the individual subject level \citep{essen2012hcp}. We use data from the 2017 release accessed through ConnectomeDB. Details on data acquisition and preprocessing pipeline of dMRI and sMRI data in the HCP can be found in \cite{essen2012hcp, glassier2013hcp,glasser2016hcp}. To produce connectome data from raw dMRI/sMRI data, we use the reproducible probabilistic tractography algorithm in \cite{girard2014tracto} to construct tractography data for each subject, the DK atlas \citep{desikan2006automated} to define the brain parcellation, and the preprocessing pipeline in \cite{zhang2018pipeline} to extract weighted matrices of connections. More details of these steps can be found in \cite{zhang2019tensor}. In the extracted data matrices, each connection is described by a scalar number. The HCP data include scores for many behavioral traits related to cognition, motor skills, substance use, sensories, emotions, personalities, and many others. Details can be found at \url{https://wiki.humanconnectome.org/display/PublicData/HCP-YA+Data+Dictionary-+Updated+for+the+1200+Subject+Release}. The final data set consists of $n = 1065$ brain connectomes and 175 traits. 

 To study relationships between the AM and tree-based representations of the brain connectome, 
 we use a simple two-step approach: first vectorizing the connectome data, and then applying a regression method. 
 Since a connectome matrix is symmetric, we only vectorize the upper triangular part, resulting in a 2278-dimensional vector for each matrix. We also remove pairs of ROIs that show no variability in connectivity across subjects, reducing the vectorized matrices to 2202 dimensions. We construct the tree representations of the connectomes based on the construction described in Section 2.1. 
 Self-edges are not recorded in the connectome matrices, so the leaves in the tree representation have zero weights. As a result, we remove the leaf nodes, leaving the vectorized trees to be 23-dimensional instead of 91-dimensional. 
 
 To apply regression algorithms, we first reduce the dimension of the adjacency matrix to $K\ll 2202$ by selecting the top principal components (PCs). We observed that the regression MSEs of commonly used algorithms including linear regression, decision trees, support vector machines (SVM), boosting, Gaussian process (GP) regression, etc, did not decrease when we increased the number of PCs, and \cite{zhang2019tensor} also observed that the regression performance is robust for $K=~20\sim60$. As a result, we keep the first $K=23$ PCs to match the dimension of the tree for a fair comparison. 
 We refer to the data from the vectorized trees as $D_{tree}$ and that from principal components of the vectorized matrices as $D_{PCA}$. Figure \ref{fig:tree-examples}a compares the visualization of the average connectomes matrix to the average connectomes tree from the HCP data. Figure \ref{fig:tree-examples}b shows the visualization of the tree representation as a more effective exploratory analysis tool since one can easily find regions of the brain with large percentage change between the top and bottom scores in four cognitive traits. 
 
\subsection{Canonical correlation analysis}

Human brain connectivity has been shown to be capable of explaining significant variation in a variety of human traits. Specifically, data on functional and structural connectivity have been used to form latent variables that are positively correlated with desirable, positive traits (i.e., high scores on fluid intelligence or oral reading comprehension) while also being negatively correlated with undesirable traits (i.e., low sleep quality, frequent use of tobacco or cannabis) \citep{smith2015cca, tian2020gradient}. In this first analysis, we compare how strongly latent variables inferred from the two representations are associated with the perceived desirability of traits. We choose to work with a subset of 45 traits that have been shown to strongly associate with brain connectome variations, including cognitive traits, tobacco/drug use, income, years of education, and negative and positive emotions \citep{zhang2019tensor, smith2015cca}. Each cognitive trait is often measured by several metrics in an assessment. We only include a primary metric that has been adjusted for the subject’s age if applicable. We include only traits with continuous values to simplify the analyses. The full list of variables used can be found in the Appendix. We will use a statistical method called canonical correlation analysis (CCA), which finds linear combinations of the predictors that are most correlated with some other linear combinations of the outcomes. In our case, the predictors are principal components of the vectorized AM or features of the vectorized trees, and the outcomes are scores measured on 45 traits. Mathematically, let $X \in \mathbb{R}^{n \times p}$ be the feature matrix ($n = 1065, p = 23$), and $Y \in \mathbb{R}^{n \times q}$ be the outcome matrix ($q=39$ with some traits being removed due to extensive missing values, see the next paragraph for more details). Additionally, let $\boldsymbol{a}_k \in \mathbb{R}^p$ and $\boldsymbol{b}_k \in \mathbb{R}^q$, $k = 1,..., \min(p, q)$ be pairs of linear transformation of the data. We refer to $(\boldsymbol{r}_k, \boldsymbol{s}_k) = (X\boldsymbol{a}_k, Y\boldsymbol{b}_k)$ as the feature and outcome canonical variates or the $k^{th}$ pair of canonical variates. CCA aims to learn $\boldsymbol{a}_k, \boldsymbol{b}_k$ such that $(\boldsymbol{a}_k, \boldsymbol{b}_k) = \argmax_{\boldsymbol{a}_k, \boldsymbol{b}_k} \text{corr}(\boldsymbol{r}_k, \boldsymbol{s}_k)$, so the linear transformations maximize the correlation between components of the canonical variates pair. Pairs of canonical variates are also constrained to be orthogonal, that is, $\boldsymbol{r}_k^T\boldsymbol{r}_h=0$ and $\boldsymbol{s}_k^T\boldsymbol{s}_h=0$ for $k \neq h$. 

Since this analysis considers all traits simultaneously, we remove traits with extensive missing values (i.e., more than 10\% of all obsevations) and are left with 39 traits. We then normalize $D_{tree}$, $D_{PCA}$ and these 39 trait scores, and fill in missing values with the feature's mean. We fit two CCAs using $D_{tree}$ and $D_{PCA}$ separately, and use the Wilks's lambda test to check that the first canonical variate pair, which has the largest co-variation, is significant at the 5\% level in each model. We hand-label the traits as desirable or not desirable based on previous research \citep{smith2015cca, tian2020gradient}. Figure \ref{fig:cca} plots the correlation of the trait scores with the first feature canonical variate and color-codes traits by desirability. Traits with smaller fonts have smaller contributions to the linear combination that makes up the first outcome canonical variate. It shows that desirable traits tend to be more highly negatively correlated with the first feature canonical variate compared to the undesirable traits, which tend to be weakly correlated. This is most clear with traits with strong signal such as fluid intelligence and spatial orientation. With the same number of features, the tree representation produces a canonical variate that can separate traits into groups of the same desirability slightly better than the principal components of the AM can.

\begin{figure}[!h]
    \centering
    \includegraphics[width=\textwidth]{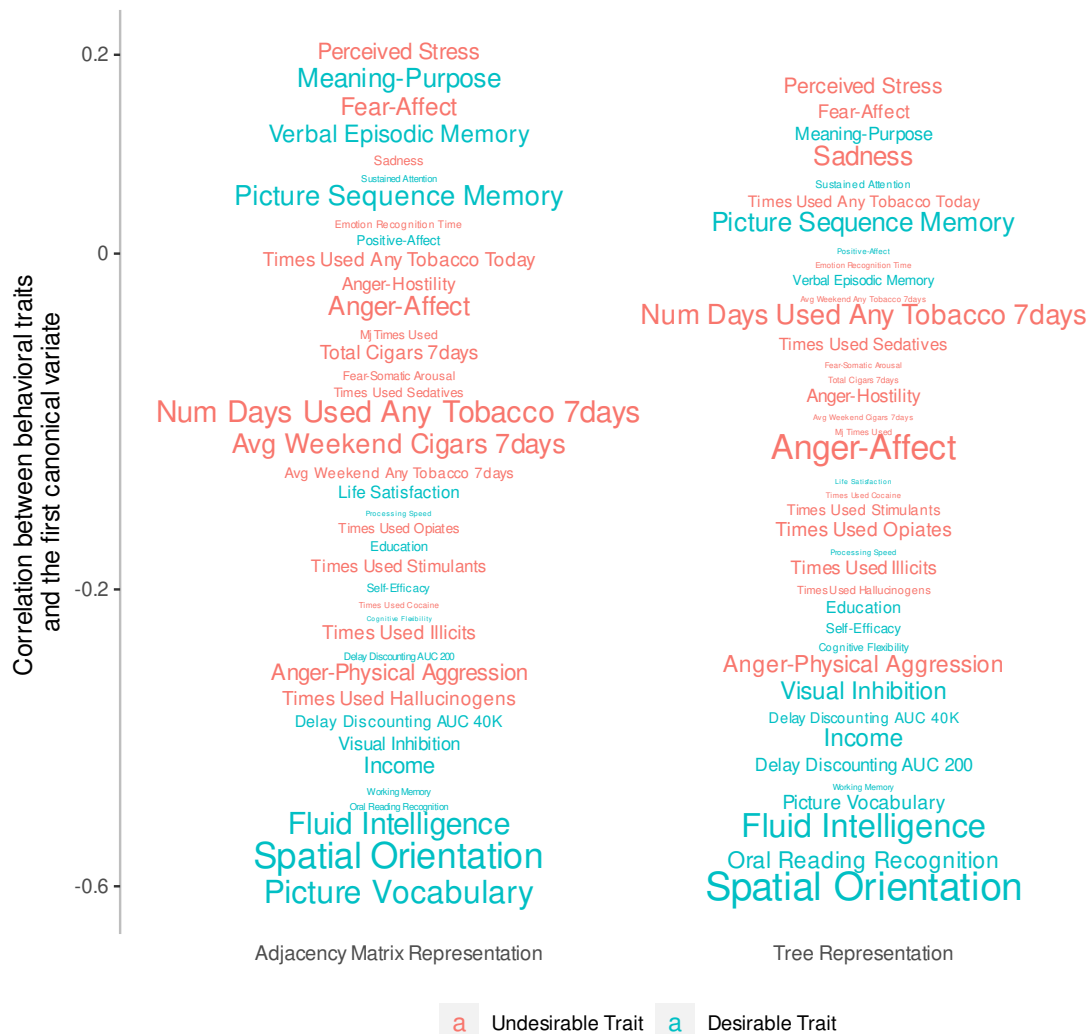}
    \caption{Correlations between behavioral traits and the first canonical variate extracted from 23 principal components of the AM compared to the 23 non-leaf nodes of the tree representation. The y-axis has been transformed so that traits do not overlap. The font size of each trait indicates the magnitude of the coefficients of a linear combination that defines the first canonical variate.}
    \label{fig:cca}
\end{figure}

\subsection{Prediction}

 Additionally, we consider the performance of these representations in predictive tasks. We hypothesize that if the tree representation preserves important information from the AM representation, they will provide comparable performance in predicting trait scores. Since in section 3.2, cognitive traits generally have the largest correlations with the first feature canonical variate, we will examine cognitive traits in more details in this section. Specifically, we include all 45 cognitive traits in the HCP data, including different metrics of the same trait. We fit a baseline model that returns the sample mean and 19 popular machine learning models (including linear regression, decision tree, SVM, ensemble trees, GP regression, and their variants) to the $D_{tree}$ and $D_{PCA}$ data. To evaluate predictive performance, we consider two scale-free metrics: 1) correlations between predictions and true outcomes and 2) the percentage of improvement in test MSE compared to the baseline predictor. We calculate these metrics using 5-fold cross validation repeated  10 times. Figure \ref{fig:MSE} shows the cross-validated predictive performance for two representative regression algorithms: linear regression and GP regression. Linear regression represents a simple, interpretable, and widely used algorithm, while GP regression is a flexible algorithm that has the best overall performance among the 19 algorithms we studied. For each algorithm and each trait, the $x$-axis is the performance of the AM representation while the $y$-axis is that of the tree representation. Points above the diagonal line $y=x$ indicate better performance by the tree representation, while points below the line indicate better performance by the AM representation.

\begin{figure}[!h]
\centering
    \includegraphics[width=\textwidth]{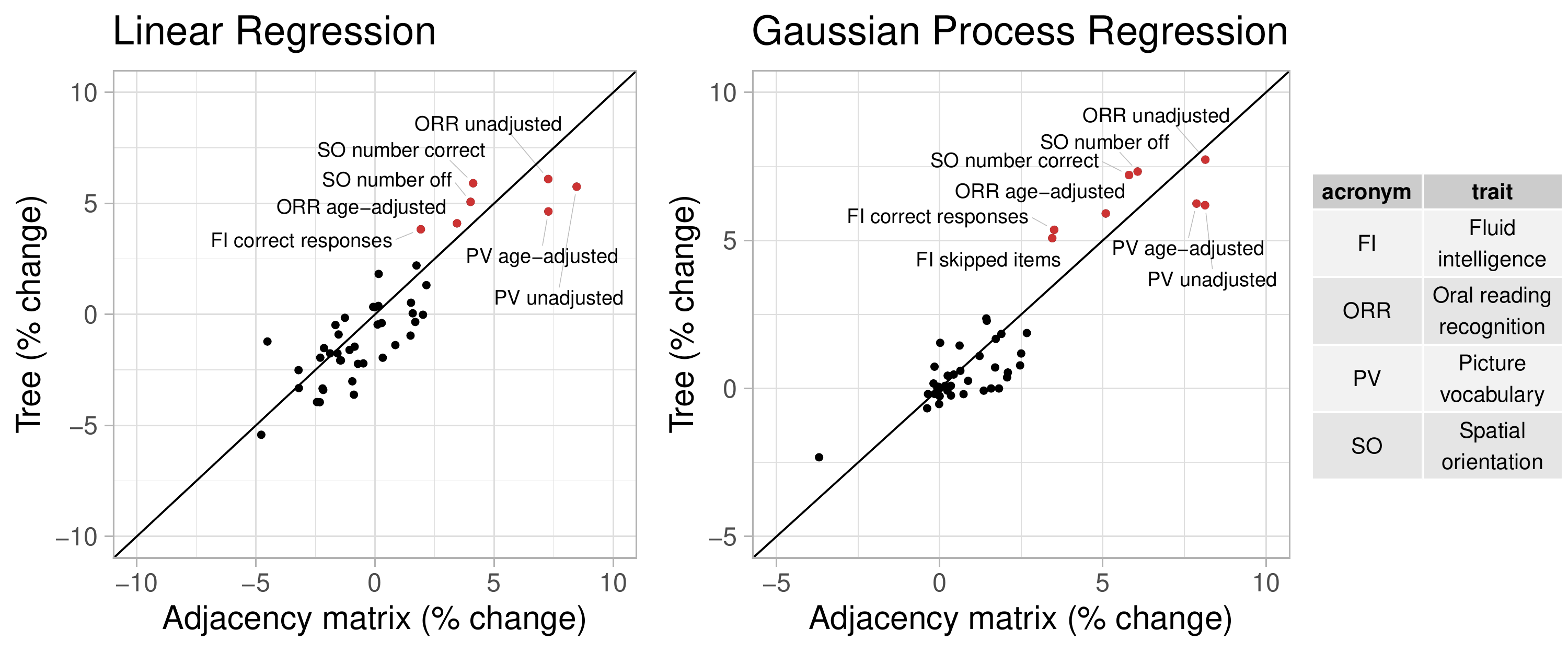}
    \caption{Percentage of change in MSE compared to baseline of linear regression and GP regression using tree and AM representation in predicting 45 cognitive traits}
    \label{fig:MSE}
\end{figure}

\begin{figure}[!h]
\centering
    \includegraphics[width=\textwidth]{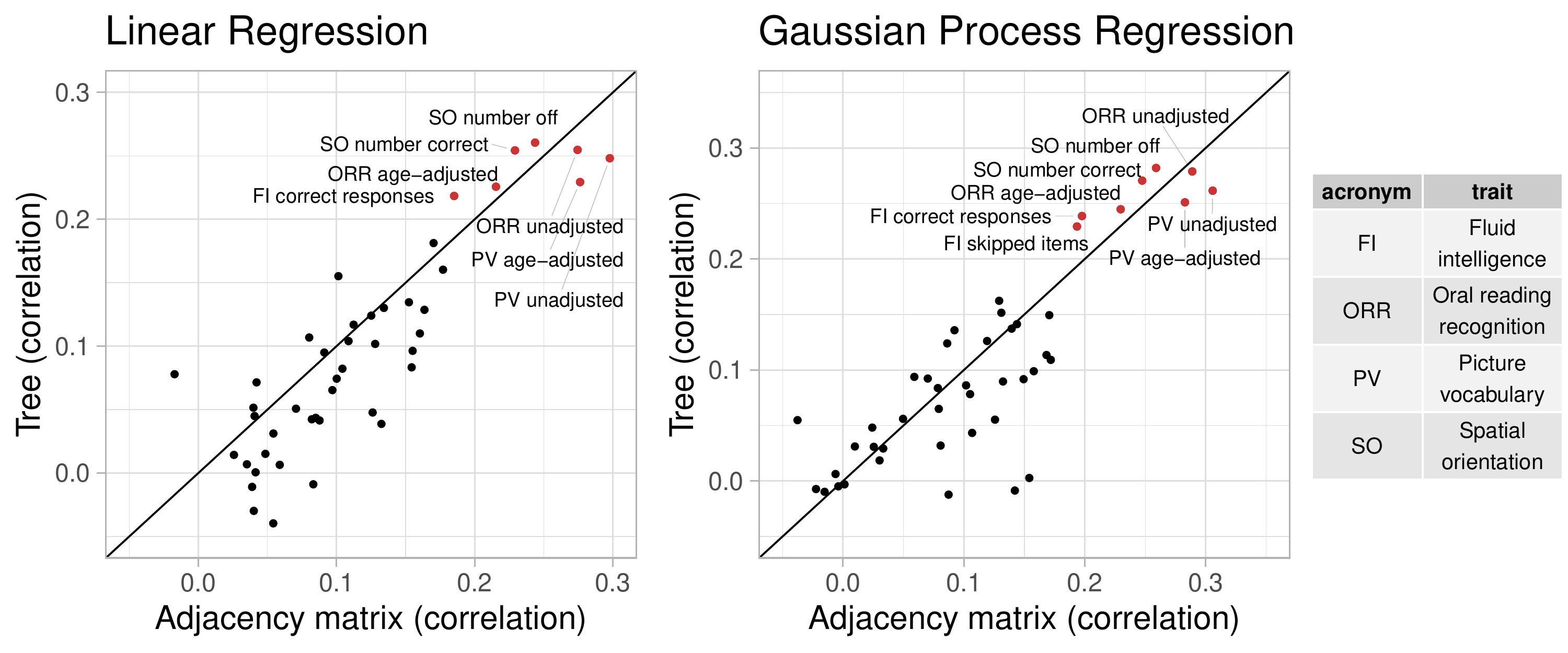}
    \caption{Correlation between predictions and observed outcomes of linear regression and GP regression using tree and AM representation in predicting 45 cognitive traits}
    \label{fig:Corr}
\end{figure}

Overall, the performance of both representations seems similar. However, for most traits, even when considering the GP regression with best overall performance, the correlation is smaller than 0.2, and the improvement in test MSE is less than 3 percent. This suggests that the vectorized brain connectivity might not be relevant to predicting most of these traits. If we focus on traits with large correlations or improvement in MSE, the tree representation has better performance in terms of both correlation and improvement in MSE for five out of eight traits with correlations greater than 0.2. These traits include fluid intelligence, picture vocabulary, spatial orientation, and oral reading recognition, which also have the largest correlations with the canonical variate in section 3.2.

\subsection{Interpretability of regression}

Finally, we compare interpretability of the two representations. Scientists are often interested in identifying structures in the connectomes that are associated with traits to answer questions such as which kind of connections might be damaged by routine use of drugs or which might be responsible for enhanced working memory. Therefore, we compare interpretability, in terms of biologically meaningful inference, of regression results using the two representations. We focus on spatial orientation (number correct), picture vocabulary (age-adjusted), fluid intelligence (correct responses), oral reading recognition (age-adjusted), and working memory (age-adjusted) because these traits show strong associations with the connectomes in the CCA and prediction tasks. 

We fit a separate linear regression with Bayesian model selection on $D_{tree}$ and $D_{PCA}$ to infer associations between brain connectomes and the traits mentioned above. Bayesian model selection accounts for uncertainty in the model selection process by posterior probabilities for the different possible models. The regression coefficients are averaged across all models, weighted by estimated model posterior probabilities. For feature selection, we define important features as those with posterior inclusion probabilities of more than 0.75. For models using $D_{tree}$, we simply interpret posterior means and credible intervals of the estimated effects of important features directly. Figure \ref{fig:infer} and Figure \ref{fig:infer2} (left column) show tree features whose colors are based on their estimated coefficients, and opacity are based on their posterior inclusion probabilities. 

For the models using $D_{PCA}$, after selecting important principal components, we can calculate a regression coefficient for each brain connection from the regression coefficients of these important principal components. Let $X$ be an $n \times p$ matrix of vectorized brain connections whose columns have been standardized. Recall that PCA is based on the singular-value decomposition of the data matrix $X = U D V^T$, where $VV^T = I, U^TU = I$ and $D$ is a diagonal matrix of $p$ non-negative singular values, sorted in decreasing order. The $j^{th}$ principal axis is the $j^{th}$ eigenvector or the $j^{th}$ column of $V$, and the $j^{th}$ principal component is the $j^{th}$ column of $UD$. With $K$ PCs, we get a rank-$K$ approximation of $X \approx U_K D_{K}V_{K}^T$ where $M_K$ contains the first $K$ columns of matrix $M$. Applying the approximation to the linear model $Y = X\boldsymbol{\beta} + \boldsymbol{\epsilon} \approx U_K D_K \boldsymbol{\theta} + \boldsymbol{\epsilon}$, we get $V_{K}^T \boldsymbol{\beta} \approx \boldsymbol{\theta}$, where $ \boldsymbol{\theta}$ is a $K$-vector of regression coefficients of the principal components. There are multiple generalized inverses $\hat{\boldsymbol{\beta}} = V_{K}  \boldsymbol{\theta} + \bf{b}$ for $\bf{b}$ such that $V_{K}^T \bf{b} = 0$. We will use $\bf{b} = 0$ to get the standard least-norm inverse $\hat{\boldsymbol{\beta}} = V_{K} \boldsymbol{\theta}$ as estimates of the regression coefficient for the original brain connections \citep{West03bayesianfactor}. We interpret the results based on 50 connections with the largest coefficient magnitude. Figure \ref{fig:infer} and Figure \ref{fig:infer2} (right column) show these connections colored by their estimated coefficients.


The tree-based model finds no structure statistically significant (with posterior inclusion probability greater than 0.75) in predicting enhanced working memory. This is consistent with previous sections showing weak signals for associations between working memory and brain connectomes. The matrix-based model finds some cross-hemispheric connections between temporal lobes positively correlated with enhanced memory, which is supported by prior research \citep{erriksson2015wmemo}. However, it also finds many right hemispheric connections negatively correlated with these scores, which contradicts some prior findings \citep{poldrack1998memory}. 
For oral reading recognition, the models find cross-hemispheric and left hemispheric connections important in both the tree and adjacency matrix representations. This is consistent with prior research that showed better reading ability associated with more cross-hemispheric connections between frontal lobes \citep{zhang2019tensor} and with increased fractional anisotropy in some left hemispheric fiber tracts \citep{yeatman2012development}. The matrix-based model additionally finds many connections across all regions in the right hemisphere to be important, among which the insula, frontal opercular, and lateral temporal lobe have previously been found to correlate with this score \citep{kristanto2020predicting}.
For fluid intelligence, the tree-based model finds connections within the left hemisphere and between hemispheres to be important. The matrix-based model finds within-hemisphere connections, especially those involving the frontal lobes, to be important. The fluid intelligence score serves as proxy for ``general intelligence" \citep{raven2000test, gershon2014nih}, which relies on many sub-networks distributed across the brain \citep{dubois2018distributed}.
For the vocabulary task, the tree-based model finds connections within the left hemisphere, while the matrix-based model finds connections within both hemispheres, to be important. Both results support existing findings that interpreting meanings of words activates many regions across the brain \citep{huth2016natural}.
Finally, in the spatial orientation task, the tree-based model finds cross-hemispheric connections to be important and positively correlated with better scores, while the matrix-based model does for connections between regions within each hemisphere (Figure \ref{fig:infer2} bottom). Both are somewhat similar to prior research that found decreased cross-hemispheric and right hemispheric connectivity to be associated with impaired spatial recognition of stroke patients \citep{ptak2020discrete}.
Overall, we observe that the results from the tree representation are easier to interpret because of the inherent low-dimensional and biologically meaningful structure. The results from the AM representation tend to be noisier, and more likely to involve negative correlations between connectivity and better performance in cognitive tasks. Potentially, the representations may better encode different kinds of information that are both important.

\begin{figure}
    \centering
    \includegraphics[height=0.9\textheight]{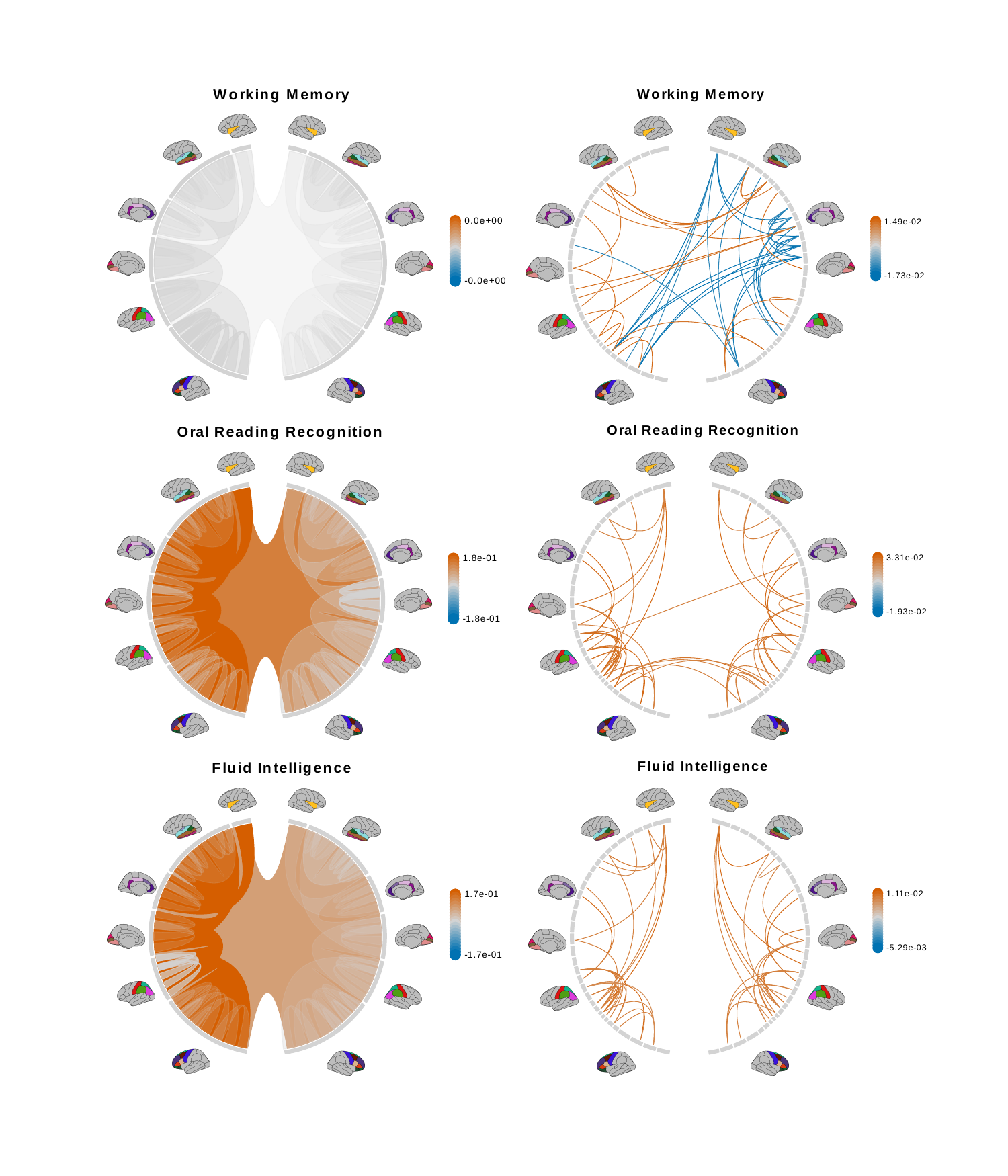}
    \caption{Brain connectome structures significantly associated with each trait inferred using the tree representation (left) and principal components of the AM representation (right). Colors represent the sign and magnitude of significant (i.e., posterior inclusion probability greater than 0.75) regression coefficients. For tree-based results, the opacity also represents the posterior inclusion probability.}
    \label{fig:infer}
\end{figure}

\begin{figure}
    \centering
    \includegraphics[height=0.6\textheight]{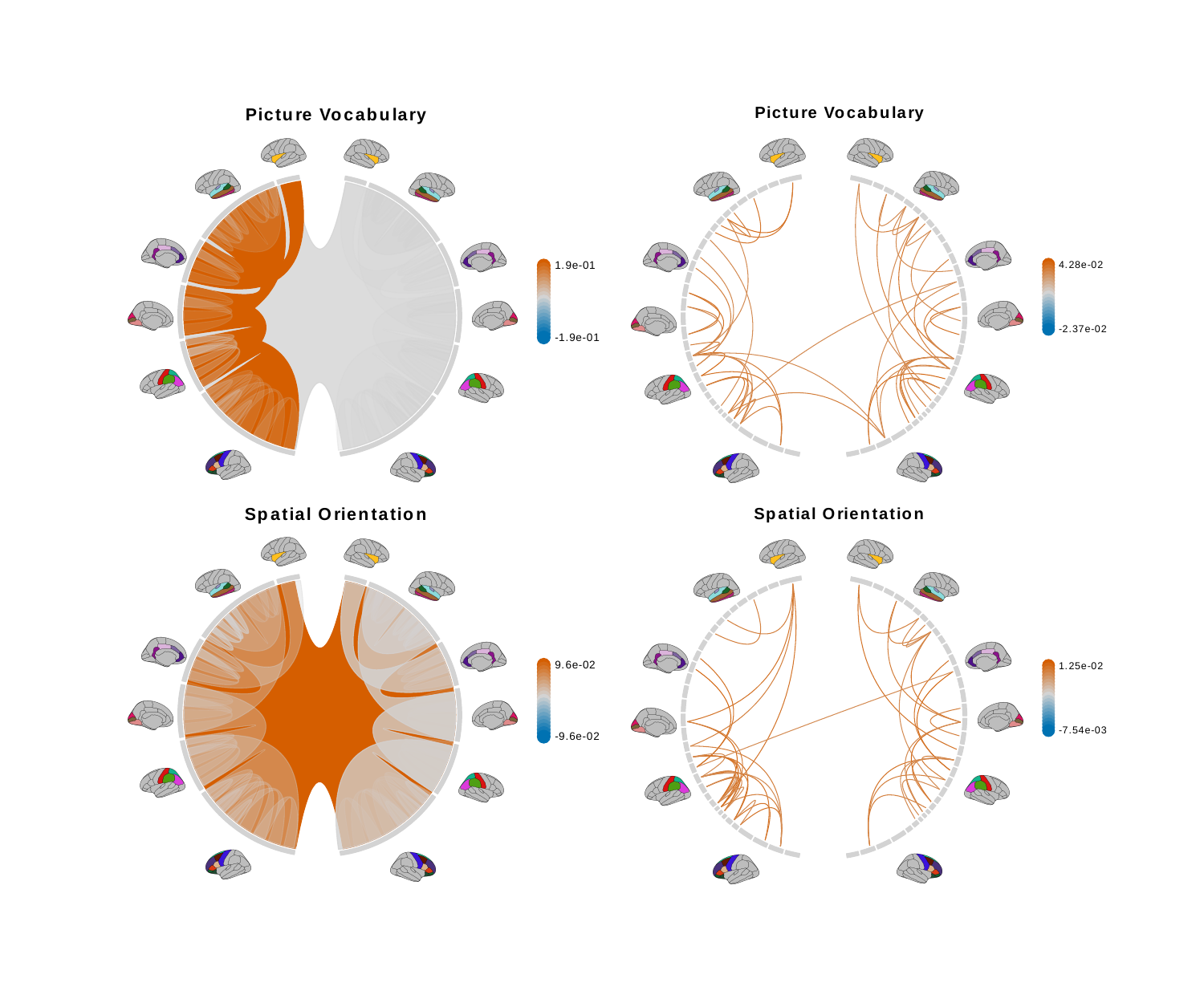}
    \caption{(Continue) Brain connectome structures significantly associated with each trait inferred using the tree representation (left) and principal components of the AM representation (right). Colors represent the sign and magnitude of significant (i.e., posterior inclusion probability greater than 0.75) regression coefficients. For tree-based results, the opacity also represents the posterior inclusion probability.}
    \label{fig:infer2}
\end{figure}

\section{Discussion}\label{sec:discussion}

We propose a novel and efficient tree representation based on persistent homology for the brain network. Through analyses of the HCP data, we show that the tree representation preserves information from the AM representation that relates brain structures to traits while being much simpler to interpret. Simultaneously, it reduces the computational cost and complexity of the analysis because of its inherent lower dimension. We believe the advantages of the tree representation will be more evident on small brain imaging data sets. 

Our new representation opens doors to new mathematical and statistical methods to analyze brain connectomes; in particular, taking into account the tree structure of the data. Topological data analysis (TDA) uses notions of shapes and connectivity to find structure in data, and persistent homology is one of the most well-known TDA methods \citep{wasserman2017tda}. TDA has been used successfully in studying brain networks \citep{saggar2018tda, gracia2020tda}, but we provide a fundamentally different approach.
Our analyses of the connectome trees in this paper are simplistic. We treat tree nodes as independent and non-interacting. Future work should consider the tree structure to enforce dependence between the nodes, and hence, between their effects on behavioral traits. The tree structure may also be exploited to model interactions between connectome structures across different scales. For instance, Bayesian treed models are flexible, nonparametric methods that have found widespread and successful applications in many domains \citep{linero2017bart}. Existing treed models might prove unwieldy to fit and interpret on AM-based brain networks but their modifications may fit the nature of our tree representation well.

\section{Appendix}

\subsection{Hierarchy in one hemisphere based on the Desikan-Killiany protocol}

\begin{figure}[ht!]
    \centering
    \includegraphics[width=0.8\linewidth, trim={0.75 3in 0.5 0.5in},clip]{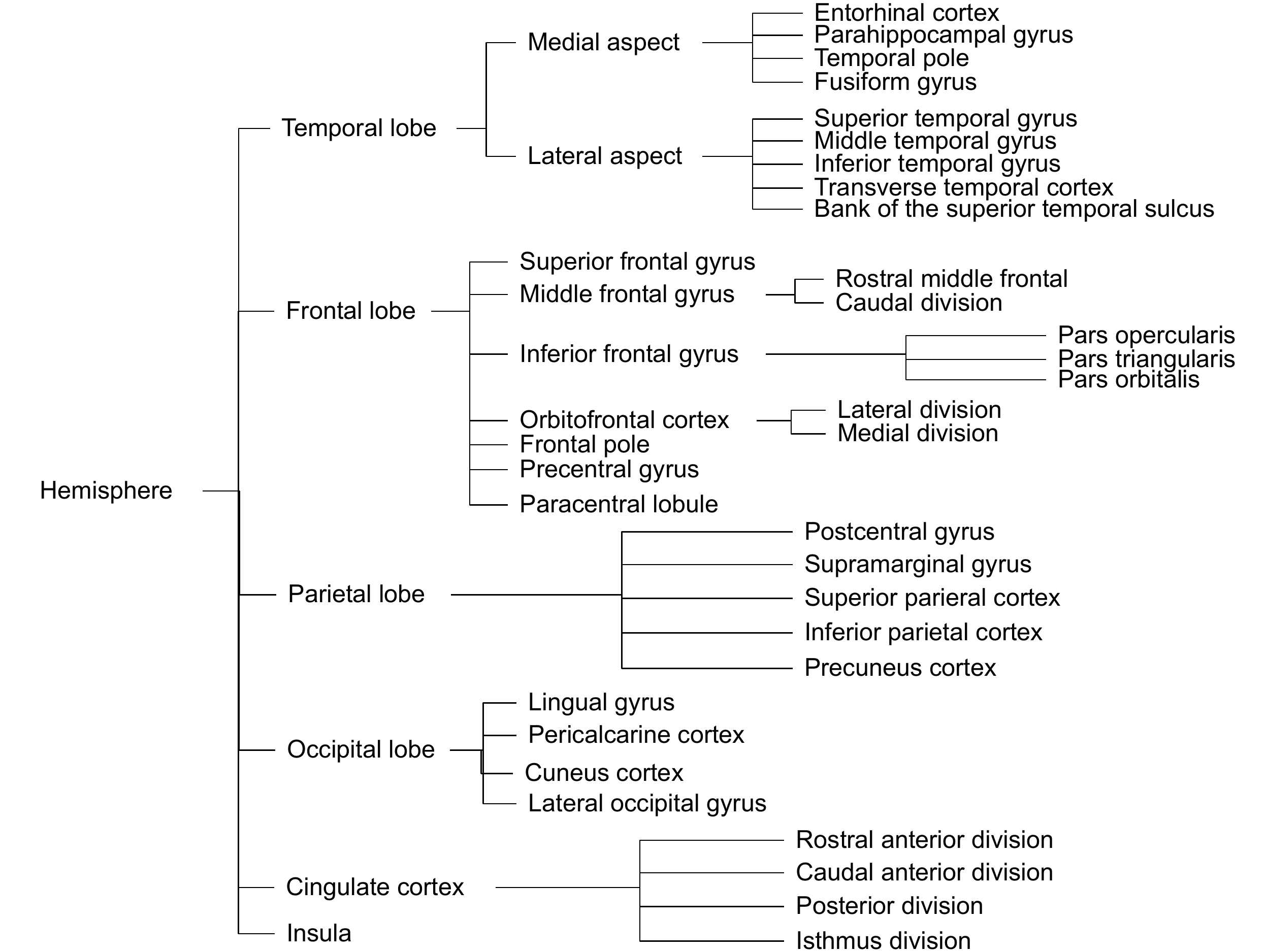}
    \label{fig:appx-brain-tree}
    \caption{The full brain hierarchy based on the Desikan-Killiany protocol for one hemisphere.}
\end{figure}

\subsection{Proof of theorem in Section 2}

We first restate Theorem \ref{thm:persistent} in topological language. Let $x^l_j\in A^l_j$ be any location in brain region $A^l_j$ and any fiber inside $A^l_j$ is treated as a loop at $x^l_j$, denoted by $c^l_{j(1)},\cdots,c^l_{j(N^l_j)}$. So topologically speaking each region $A^l_j$ consists of a point ($0$-cell) and $N^l_j$ loops ($1$-cell). Assume $A^l_i$ and $A^l_j$ are children of $A^{l-1}_k$, then we have a inclusion map from the children to the parent: $F^l_k: A^l_i\cup A^l_j\rightarrow A^{l-1}_k $, $F(x^l_j)=F(x^l_i)=x^{l-1}_k$, so $F^l_k$ induces a homomorphism between chain groups:
$F^l_k: C_p(A^l_i\cup A^l_j)\rightarrow C_p(A^{l-1}_k)$, $p=0,1$. Observe that $F$ maps all loops $c^l_{j(k)}$ and $c^l_{i(k)}$ to loops in $A^{l-1}_k$ while any fiber (path) connecting $A^l_i$ and $A^l_j$ is also mapped to loops in $A^{l-1}_k$. 
\begin{manualtheorem}{\ref{thm:persistent}'}\label{thm':persistent}
	$H^l_k$ is the corank of the persistent homology ${F^l_{k}}_*\left(H_p(A^l_i\cup A^l_j)\right)\subset H_p(A^{l-1}_k)$.	
\end{manualtheorem}

Before proving the theorem, we prove the following useful lemma:
\begin{lemma}\label{chainmap}
	$F$ defined above is a chain map so it induces a homomorphism between homology groups:
	$${F^l_{k}}_*: H_p(A^l_i\cup A^l_j)\rightarrow H_p(A^{l-1}_k),\;p=0,1.$$
\end{lemma}
\begin{proof}
	Since the chain group is nontrivial for $p=0$ and $p=1$ only, it suffices to check $\partial\circ F^l_k(c)=F^l_k\circ \partial (c)$ for any $c\in C_1(A^l_i\cup A^l_j)$. 
	
	If $c$ is a loop, then $F^l_k(c)$ is also a loop by the construction of $F^l_k$, so we have 
	$\partial\circ F(c)=x^{l-1}_k-x^{l-1}_k=0$,  $F^l_k\circ \partial (c)=F^l_k(0)=0$. 
	
	If $c$ is not a loop, we assume $c$ starts from $x^l_j$ ends at $x^l_i$, so $F^l_k(c)$ is still a loop, so we have
	$\partial\circ F^l_k(c)=x^{l-1}_k-x^{l-1}_k=0$. Since $F^l_k\circ \partial (c)=F^l_k(x^l_i-x^l_j)=x^{l-1}_k-x^{l-1}_k=0$, we conclude that $\partial\circ F^l_k=F^l_k\circ \partial$.
\end{proof}
Now we can prove Theorem \ref{thm':persistent}:
\begin{proof}[Proof of Theorem \ref{thm':persistent}]
	From the proof of Lemma \ref{chainmap}, $\rank(H_p(A^{l-1}_k))$ is the sum of the number of fibers in $A^l_i$ and $ A^l_j$ as well as the number of fibers connecting $A^l_i$ and  $A^l_j$. Since all homologous classes in $H_p(A^l_i\cup A^l_j)$ are mapped to nontrivial homologous classes in $H_p(A^{l-1}_k)$, $\rank\left({F^l_{k}}_*\left(H_p(A^l_i\cup A^l_j)\right)\right)=\rank \left(H_p(A^l_i\cup A^l_j)\right)$, which is equal to the number of fibers in fibers in $A^l_i$ and $ A^l_j$. As a result, 
	$$\corank\left({F^l_{k}}_*\left(H_p(A^l_i\cup A^l_j)\right)\right)=\rank(H_p(A^{l-1}_k))-\rank\left({F^l_{k}}_*\left(H_p(A^l_i\cup A^l_j)\right)\right)=H^{l}_k.$$
\end{proof}

\subsection{Table of 45 traits used in section 3.2 analysis}

Table \ref{tab:appx-hcp-variable-names} lists the different traits used in our analyses grouped into categories; we provide the names used in the HCP data files along with a link to look up a detailed definition of each of these traits.

\begin{table}[h!]
    \centering
    \begin{tabular}{|m{3cm}|m{12cm}|}
        \hline
        Trait category & HCP column names \\ 
        \hline
        Cognition & PMAT24\_A\_CR, ReadEng\_AgeAdj, PicVocab\_AgeAdj, IWRD\_TOT, ProcSpeed\_AgeAdj, DDisc\_AUC\_200, DDisc\_AUC\_40K, VSPLOT\_TC, SCPT\_TPRT, ListSort\_AgeAdj, PicSeq\_AgeAdj, SSAGA\_Educ, SSAGA\_Income, CardSort\_AgeAdj, Flanker\_AgeAdj\\ 
        \hline
        Emotion &  ER40\_CRT, AngAffect\_Unadj, AngHostil\_Unadj, AngAggr\_Unadj, FearAffect\_Unadj, FearSomat\_Unadj, Sadness\_Unadj, PercStress\_Unadj, SelfEff\_Unadj, LifeSatisf\_Unadj, MeanPurp\_Unadj, PosAffect\_Unadj\\ 
        \hline
        Tobacco use &  SSAGA\_TB\_Age\_1st\_Cig,  SSAGA\_TB\_DSM\_Difficulty\_Quitting, SSAGA\_TB\_Max\_Cigs,  SSAGA\_TB\_Reg\_CPD, SSAGA\_TB\_Yrs\_Smoked, Times\_Used\_Any\_Tobacco\_Today,  Avg\_Weekend\_Any\_Tobacco\_7days, Total\_Cigars\_7days, Avg\_Weekend\_Cigars\_7days, Num\_Days\_Used\_Any\_Tobacco\_7days\\
        \hline
        Drug use & SSAGA\_Times\_Used\_Illicits, SSAGA\_Times\_Used\_Cocaine,  SSAGA\_Times\_Used\_Hallucinogens, SSAGA\_Times\_Used\_Opiates,  SSAGA\_Times\_Used\_Sedatives, SSAGA\_Times\_Used\_Stimulants,  SSAGA\_Mj\_Age\_1st\_Use, SSAGA\_Mj\_Times\_Used\\
        \hline
    \end{tabular}
    \caption{Column names from the HCP data file of traits used in our analysis, grouped by categories defined by their meanings. Since there are many metrics for the same trait, we provide the column names so that their exact definitions can be looked up at \url{https://wiki.humanconnectome.org/display/PublicData/HCP-YA+Data+Dictionary-+Updated+for+the+1200+Subject+Release}.}
    \label{tab:appx-hcp-variable-names}
\end{table}

\section*{Acknowledgement}
This work was partially supported by NIH grant R01-MH118927 by the United States National Institute of Mental Health and P30ES010126 by National Institute of Environmental Health Sciences.

\newpage
\bibliographystyle{apalike}
\bibliography{ref.bib}
	
\end{document}